\newif\ifsubmit
\submittrue 
\newif\iffinal 
\finalfalse

\documentclass[letterpaper,11pt]{article}
\usepackage{typearea}
\typearea{14}

\newcommand{\CC}{{\mathcal C}}

\usepackage{typearea}
\typearea{15}

\usepackage{latexsym,graphicx}
\usepackage{amsmath,amssymb}
\usepackage{amsthm,enumerate}
\usepackage{gensymb}
\usepackage{xspace}
\usepackage{bm}
\usepackage{ifpdf}
\usepackage{xcolor}
\usepackage{algorithm}
\usepackage[noend]{algpseudocode}
\usepackage{verbatim}
\usepackage{mathrsfs}
\usepackage{paralist}
\usepackage{times}
\usepackage{nicefrac}

\allowdisplaybreaks

\definecolor{Darkblue}{rgb}{0,0,0.4}
\definecolor{Brown}{cmyk}{0,0.81,1.,0.60}
\definecolor{Purple}{cmyk}{0.45,0.86,0,0}

\usepackage[pdftex]{hyperref}

\newtheorem{theorem}{Theorem}[section]
\newtheorem{definition}[theorem]{Definition}

\newtheorem{lemma}[theorem]{Lemma}
\newtheorem{fact}[theorem]{Fact}

\newtheorem{claim}[theorem]{Claim}

\newtheorem{corollary}[theorem]{Corollary}

\newtheorem{assumption}[theorem]{Assumption}
\numberwithin{algorithm}{section}

\newenvironment{subproof}[1][\proofname]{%
  \begin{proof}[#1]%
}{%
  \end{proof}%
}

\newcommand{\junk}[1]{}
\newcommand{\ignore}[1]{}
\newcommand{\transpose}{\intercal}
\newcommand{\R}[0]{{\ensuremath{\mathbb{R}}}}

\def\ip#1{\langle #1 \rangle}

\usepackage{mathtools}

\DeclarePairedDelimiterX{\infdivx}[2]{(}{)}{%
  #1\;\delimsize\|\;#2%
}

\newcommand{\conv}{\operatorname{Conv}}

\newcommand{\sse}{\subseteq}

\newcommand{\calA}{{\mathscr{A}}}
\newcommand{\F}{{\mathscr{F}}}
\newcommand{\calE}{{\mathcal{E}}}
\newcommand{\calF}{{\mathscr{F}}}

\newcommand{\E}{{\mathbb{E}}}

\newcommand{\e}{\varepsilon}
\newcommand{\eps}{\varepsilon}

\usepackage[colorinlistoftodos,textsize=tiny,textwidth=2cm,color=red!25!white,obeyFinal]{todonotes}
\iffinal
	\newcommand{\agnote}[1]{}
	\newcommand{\rnote}[1]{}
	\newcommand{\mnote}[1]{}
	\newcommand{\alert}[1]{}
\else
	\newcommand{\agnote}[1]{\todo[color=blue!25!white]{AG: #1}\xspace}
	\newcommand{\rnote}[1]{\todo[color=green!25!white]{RM: #1}\xspace}
	\newcommand{\mnote}[1]{\todo[color=red!25!white]{MM: #1}\xspace}
	\newcommand{\alert}[1]{{\color{red} #1}}
\fi

\newcommand{\initOneLiners}{%
    \setlength{\itemsep}{0pt}
    \setlength{\parsep }{0pt}
    \setlength{\topsep }{0pt}
}
\newenvironment{OneLiners}[1][\ensuremath{\bullet}]
    {\begin{list}
        {#1}
        {\initOneLiners}}
    {\end{list}}

\newcommand{\squishlist}{
 \begin{list}{$\bullet$}
  { \setlength{\itemsep}{0pt}
     \setlength{\parsep}{3pt}
     \setlength{\topsep}{3pt}
     \setlength{\partopsep}{0pt}
     \setlength{\leftmargin}{1.5em}
     \setlength{\labelwidth}{1em}
     \setlength{\labelsep}{0.5em} } }

\newcommand{\squishend}{
  \end{list}  }

\newcounter{asidecounter}

\newcommand{\Opt}{\ensuremath{\mathsf{opt}\xspace}}
\newcommand{\OPT}{\ensuremath{\mathsf{opt}\xspace}}
\newcommand{\FOPT}{\ensuremath{\mathsf{fopt}\xspace}}

\newcommand{\profit}{\ensuremath{\pi}}

\newcommand{\occ}[1]{\mathsf{occ}_{{#1}}}

\newcommand{\xocc}{x^{occ}}
\newcommand{\xrest}{x^{rest}}

\newcommand{\Isample}{U^{\text{sample}}}

\newcommand{\Ioos}{U^{\oos}}

\newcommand{\calC}{\mathcal{C}}
\newcommand{\calP}{\mathcal{P}}
\newcommand{\be}{\mathbf{e}}

\newcommand{\gr}{\nabla}
\newcommand{\gs}{g^\star}
\newcommand{\fs}{f^\star}

\newcommand{\ghat}{\widehat{g}}
\newcommand{\ghats}{\widehat{g}^\star}

\newcommand{\lhat}{{\widehat{\lambda}}}

\newcommand{\Aopen}{U^{\degree}}

\newcommand{\submod}{\textsf{SubmodMS}}

\DeclareMathOperator{\argmax}{argmax}

\newcommand{\UB}{{\color{red} UB}}

\renewcommand{\paragraph}[1]{\smallskip \noindent \textbf{#1}}

\begin{document}

\title{Maximizing Profit with Convex Costs in the Random-order Model\footnote{This work was done in part while the authors were visiting the Simons Institute for the Theory of Computing.}}

\author{Anupam Gupta \and Ruta Mehta \and Marco Molinaro}

\date{}

\maketitle

\begin{abstract}
  Suppose a set of requests arrives online: each request gives some value $v_i$ if accepted, but requires using some amount of each of $d$
  resources. Our cost is a convex function of the vector of total
  utilization of these $d$ resources. Which requests should be accept to
  maximize our profit, i.e., the sum of values of the accepted demands,
  minus the convex cost?

  We consider this problem in the random-order a.k.a.\
  secretary model, and show an $O(d)$-competitive algorithm for the case
  where the convex cost function is also \emph{supermodular}. If the set
  of accepted demands must also be independent in a given matroid, we
  give an $O(d^3 \alpha)$-competitive algorithm for the supermodular
  case, and an improved $O(d^2\alpha)$ if the convex cost function is also
  separable. Here $\alpha$ is the competitive ratio of the best
  algorithm for the submodular secretary problem.
  These extend and improve previous results known for this problem.
  Our techniques are simple but use powerful ideas from convex duality,
  which give clean interpretations of existing work, and allow us to
  give the extensions and improvements.
\end{abstract}



\section{Introduction}
\label{sec:introduction}

The problem we consider is a basic convex optimization problem in the
online setting: $n$ items appear one-by-one. Each item/element $e$ has a
\emph{$d$-dimensional size} $s(e) \in \R_+^d$ and a \emph{value}
$v(e) \in \R_+$, which are both revealed to us when the item
arrives.
We must either accept or reject
an item when it arrives, before seeing the future items.  If we accept a
certain subset $A \sse [n]$ of the items, we get their total value
$v(A) := \sum_{e \in A} v_e$, but incur a production cost
$g(s(A)) := g(\sum_{e \in A} s(e))$, where $g: \R_+^d \to \R_+$ is a
non-decreasing \emph{convex cost} function with $g(0) = 0$.  Optionally,
we may also be given a downwards-closed family of subsets
$\calF \sse 2^{[n]}$, and now the accepted set of elements $A$ must lie
in $\calF$. More formally, we want to solve
\begin{gather}
  \max_{A \in \calF} \text{ profit } \profit(A) := \big[ v(A) - g(s(A))
  \big]. \label{eq:main1} 
\end{gather}
This question arises, e.g., when we are selling some service that
depends on $d$ commodities, where the value is the amount of money customer
$e$ is willing to pay for the service, and the size vector $s(e)$ is the
amount of resources she will require. The cost function $g(\cdot)$
captures our operating expenses; its convexity models \emph{diseconomies
  of scale} that arise when dealing with scarce commodities.
 In particular, it can capture $d$-dimensional knapsack constraints, by setting $g(z) = 0$ until the knapsack size, and $\infty$ afterwards. When the cost function is linear $g(z) = \ip{a,z}$, we want to pick a max-weight subset from $\calF$ using
item weights $v(e) - \ip{a, s(e)}$, which is tractable/approximable for $\calF$ being a matroid, $p$-system, etc.

Blum et al.~\cite{BGMS11} defined this problem in the adversarial model,
and gave posted-price algorithms for ``low-degree'' \emph{separable} cost functions $g$, that is, of the form $g(z) = \sum_{i=1}^d g_i(z_i)$ for 1-dimensional functions $g_i$'s. This result was tightened by Huang and Kim~\cite{HK15}, still for separable functions with additonal growth control. More recently, Azar et al.~\cite{AzarBCCCGHKNNP16} studied this problem for more general \emph{supermodular} non-separable convex functions $g$ (see also~\cite{EF16}). A differentiable function $g$ is supermodular if for any vectors $x \le x'$ we have $\gr g(x) \le \gr g(x')$. Equivalently, if $g$ is
twice-differentiable, it is supermodular if
$\frac{\partial^2 g}{\partial x_i \partial x_j} \geq 0$ for all
$i \neq j$, i.e., increasing the consumption of
a resource cannot decrease the marginal cost for another. However, to
handle the worst-case ordering, Azar et al.\ also require the
cost functions to have essentially low-degree.

Can we do better by going beyond the worst-case model? In this paper, we focus on the random-order or
``secretary'' setting, where the set of items is fixed by an adversary
but they arrive in random order. In the single-dimensional case $d=1$,
it is easy to see that a solution that learns a ``good'' threshold $\lambda$ and picks all further items with density $v(e)/s(e)$ at least $\lambda$ essentially gives a constant approximation, much like in the secretary and knapsack secretary problems~\cite{freeman-secretary,babaioff}. The multi-dimensional case is much more challenging. This was studied by Barman et al.~\cite{BUCM}, again assuming a separable cost function $g(z) = \sum_{i=1}^d g_i(z_i)$. They give an $O(d)$-competitive algorithm for the unconstrained case, and an $O(d^5 \alpha)$-competitive algorithm for the problem with a downward closed constraint set $\F$, where $\alpha$ is the competitive ratio
for the $\F$-secretary problem. Their main idea is to perform a clever
decomposition of the value of each item into ``subvalues'' $v_i(e)$ for
each of the coordinate cost functions $g_i$'s; this effectively
decomposes the problem into $d$ 1-dimension problems with values $v_i$'s
and costs $g_i$'s. Unfortunately, since their solution explicitly relies on the decomposability of the cost function, it is unclear how to extend it to general supermodular functions. 
We note that when the cost function is supermodular, the profit function
is a \emph{submodular} set function (Section
\ref{sec:superm-funct}). However, the profit can take \emph{negative
  values}, and then existing algorithms for submodular maximization
break down.\footnote{For example, we can model set packing (which is
  $\Omega(\sqrt{\textrm{\# sets}})$-hard) as follows: for a
  subcollection $\mathcal{S}$ of sets, let $\pi(\mathcal{S}) =
  |\bigcup_{S \in \mathcal{S}} S| - \sum_{S \in \mathcal{S}} (|S| -
  1)$. The function $\pi$ is submodular, and its maximizer is a largest
  set packing.}

Our work is then motivated by trying to better understand the multi-dimensional nature of this problem, and provide a more principled algorithmic approach.



\subsection{Our Results}
\label{sec:our-results}

We use techniques from convex duality to re-interpret, simplify, and improve the existing results. First, we obtain the first approximation for non-separable supermodular cost functions. (We omit some mild regularity conditions for
brevity; see Section~\ref{sec:unconstr} for full details.)

\begin{theorem}[Unconstrained \& Supermodular]
  \label{thm:main1}
  For the unconstrained problem with \emph{supermodular} convex cost
  functions $g$, we give an $O(d)$-competitive randomized algorithm in
  the random-order model.
\end{theorem}

This result generalizes the $O(d)$-approximation of Barman et
al.~\cite{BUCM} to the non-separable case. The factor $d$ seems unavoidable, since our problem inherits the (offline) $\Omega(d^{1-\e})$ hardness of the $d$-dimensional knapsack, assuming $NP \neq ZPP$~\cite{DGV05}.

Next, we consider the constrained case. For simplicity, we focus on the
most interesting case where $\calF$ is a matroid constraint; more
general results can be obtained from the results and techniques in
Section~\ref{sec:contr-online}.

\begin{theorem}[Constrained \& Separable]
  \label{thm:main2}
  For the constrained problem with $\F$ being a matroid constraint, and
  the cost function $g$ being \emph{separable}, we get an
  $O(d^2 \log \log \text{rank})$-competitive randomized algorithm in the
  random-order model.
\end{theorem}

This improves by a factor of $d^3$ the $O(d^5 \log \log \text{rank})$-approximation given
by~\cite{BUCM}.
Finally, we give a general reduction that takes an algorithm for
\emph{separable} functions and produces an algorithm for
\emph{supermodular} functions, both with respect to a matroid
constraint. This implies:

\begin{theorem}[Constrained \& Supermodular]
  \label{thm:main2b}
  For the constrained problem with $\F$ being a matroid constraint, and
  the cost function $g$ being \emph{supermodular}, we get an
  $O(d^3 \log \log \text{rank})$-competitive randomized algorithm in the
  random-order model.
\end{theorem}

On conceptual contributions are in bringing techniques from convex
duality to obtain, in a principled way, \emph{threshold-based}
algorithms for non-linear secretary problems. Since this is a classical
and heavily used algorithmic strategy for secretary problems~\cite{freeman-secretary,babaioff,kleinberg,agrawal,mRavi} we
hope that the perspectives used here will find use in other contexts.

\subsection{Other Related Work}
\label{sec:related-work}

There is a vast literature on secretary
problems~\cite{freeman-secretary}. Closest to our setting, 
Agrawal and Devanur study an online convex optimization problem in the
random order model, and give a powerful result showing strong regret
bounds in this setting~\cite{AD15}. They extend this result to give
algorithms for online packing LPs with ``large'' right-hand
sides. However, it is unclear how to use their algorithm to obtain results in our setting. Other algorithms solving packing LPs with large right-hand sides appear in~\cite{agrawal,Devanur11,mRavi,KRTV14,GM16,ESF14}.

Feldman and Zenklusen~\cite{FZ15} show how to transform any algorithm
for (linear) matroid secretary into one for \emph{submodular} matroid
secretary. They give an $O(\log \log \text{rank})$-algorithm for the
latter, based on results of~\cite{Lachish14,FSZ15}. All these algorithms critically 
assume the submodular function is non-negative everywhere, which is not
the case for us, since picking too large a set may cause the profit
function to go negative. Indeed, one technical contribution is a procedure for making the profit function non-negative  while preserving submodularity 
(Section~\ref{sec:monoton}), which allows us to use these results as part of our solution.


\subsection{Structure of the paper}
 Section~\ref{sec:unconstr} develops the convex duality perspective used in the paper for the offline version of the unconstrained case, hopefully in an manner accessible to non-experts. Section~\ref{sec:constr} gives the small changes required to extend this to the constrained case. Section~\ref{sec:contr-online} shows how transform these into online algorithms. Section~\ref{sec:separ-vs-super} shows how to convert an algorithm for separable functions into one for supermodular functions, both subject to matroid constraints. To improve the presentation, we make throughout convenient assumptions, which are discharged in Appendix~\ref{sec:other-loose-ends}.

Since some familiarity with convex functions and conjugates will be useful, we give basic facts about them and some
probabilistic inequalities in Appendix~\ref{sec:app-facts}.

\section{Preliminaries}
\label{sec:preliminaries}




\paragraph{Problem Formulation.} Elements from a universe $U$ of size $n$ are presented in random order.
Each element $e$ has value $v(e) \in \R_+$ and size $s(e) \in \R_+^d$.
We are given a convex cost function $g: \R_+^d \to \R_+$. On seeing each
element we must either accept or discard it. A downwards-closed
collection $\calF \sse 2^U$ of feasible sets is also given. When
$\calF = 2^U$, we call it the \emph{unconstrained} problem. The goal is
to pick a subset $A \in \calF$ to maximize the \emph{profit}
\begin{gather}
  \profit(A) := \sum_{e \in A} v(e) - g\big( \sum_{e \in A} s(e)
  \big). \label{eq:1}
\end{gather}
We often use vectors in $\{0,1\}^n$ to denote subsets of $U$; $\chi_A$
denotes the indicator vector for set $A$. Hence, $\calF \sse \{0,1\}^n$
is a down-ideal on the Boolean lattice, and we can succinctly write our
problem as
\begin{gather}
  \max_{x \in \calF} ~~\profit(x) := \ip{ v, x } - g(Sx), \label{eq:2}
\end{gather}
where columns of $S \in \R^{d \times n}$ are the item sizes. Let $\OPT$
denote the optimal value.  For a subset $A\subseteq U$, $v(A)$
and $s(A)$ denote $\sum_{e \in A} v(e) = \ip{v, \chi_A}$ and
$\sum_{e \in A} s(e) = S\chi_A$ respectively.


\begin{definition}[Exceptional]
  Item $e \in U$ is \emph{exceptional} if
  $\arg\max_{\theta \in [0,1]}\big\{ \theta\, v(e) - g(\theta\, s(e))
  \} \in (0,1)$.
\end{definition}


\begin{definition}[Marginal Function]
  \label{def:marg}
  Given $g: \R^d \to \R$, define the $i^{th}$ \emph{marginal function} $g_i: \R
  \to \R$ as $g_i(x) := g(x \be_i)$, where $\be_i$ is the $i^{th}$
  standard unit vector. 
\end{definition}

\ifsubmit
\else

\subsection{Convex Functions}
\label{sec:convex}

To avoid degenerate conditions, we assume that the convex cost functions
$g$ we consider are closed, not identically $+\infty$ and there is an
affine function minorizing $g$ on $\R^d$.

\begin{definition}[Convex Dual]
  \label{def:dual}
  For any function $g : \R^d \rightarrow \R$, its \emph{convex dual} is
  the function $\gs : \R^d \rightarrow \R$ given by
  $$\gs(y) := \sup_x \big[\ip{y,x} - g(x) \big].$$
\end{definition}

\begin{claim}[Linearization]
  \label{clm:linear} (\cite[Theorem~E.1.4.1]{HUL}) For every convex function $g: \R^d \to \R$, any
  (sub)gradient at the point $x$ gives the ``right linearization'':
  \begin{gather}
    g(x)  = \ip{ x, u } - \gs(u) \iff u \in \partial g(x).
  \end{gather}
\end{claim}

\begin{claim}[Double Dual]
  (\cite[Corollary~E.1.3.6]{HUL}) \label{clm:doubleDual} Let
  $g : \R^d \rightarrow \R$ be a convex function. If its epigraph
  $\{(x,r) \in \R^d \times \R : r \ge g(x) \}$ is closed, then
  $g^{\star \star} = g$.
\end{claim}

\begin{claim}[Fenchel-Young Inequality]
  \label{clm:fy-ineq}
  For every convex function $g: \R^d \to \R$, linearizing using
  any vector gives us an underestimate on the value of $g$:
  \begin{gather}
    g(x) \geq \ip{ x, u } - \gs(u).
  \end{gather}
\end{claim}

\begin{claim}[Dual Function Facts] 
  \label{clm:fenchel-props}
  Let $g: \R_+^d \to \R$ be a convex function on the positive
  orthant.
  \begin{OneLiners}
  \item[a.] If $g(0) = 0$, then $\gs(\lambda) \geq 0$ for all $\lambda$.
  \item[b.] If $\lambda \geq \lambda'$ then $\gs(\lambda) \geq
    \gs(\lambda')$.
  \item[c.] $\gs$ is a closed convex function.
  \end{OneLiners}
\end{claim}

\begin{proof}
  For property~(a), $\gs(\lambda) \geq \ip{ \lambda, 0} - g(0) = 0$
  using Claim~\ref{clm:fy-ineq}.

  For property~(b), take any $x$ in $\R^d_+$ (the domain of $g$) and
  observe $$\ip{\lambda,x} - g(x) \ge \ip{\lambda',x} - g(x).$$ Take the
  supremum over all such $x$'s in the left-hand side and use
  Definition~\ref{def:dual} to get
  $$\gs(\lambda) \ge \ip{\lambda',x} - g(x)$$ for all $x \in
  \R^d_+$. To complete the argument, take the supremum on the
  right-hand side.

  For property~(c), see \cite[Theorem~E.1.1.2]{HUL}.
\end{proof}


\begin{claim}[Duals and Marginals Commute]
  \label{clm:dual-marginal}
  Given a monotone convex $g: \R^d \to \R$, $(\gs)_i(z) =
  (g_i)^\star(z)$ for all $z \in \R$. I.e., the marginal of the dual is
  the same as the dual of the marginal.
\end{claim}

\begin{proof}
  $(\gs)_i(z) = \gs(z \be_i) = \max_x \left(\ip{z \be_i, x} -
      g(x)\right) \stackrel{g \textrm{ increas.}}{=} \max_{x_i}
    \left(z x_i - g(x_i \be_i)\right)$ $= \max_{x_i}
    \left(z x_i - g_i(x_i)\right)$ $= (g_i)^\star(z)$. This means there
    are no concerns of ambiguity when we write $\gs_i(z)$.
\end{proof}

\begin{claim}[Subadditivity over Coordinates] \label{claim:subaddCoord}
  Given a superadditive convex function $g : \R^d \rightarrow \R$, 
  \begin{align*}
    \gs(\lambda) \le \sum_i \gs_i(\lambda_i) ~~~~~\forall \lambda.
  \end{align*}
\end{claim}

\begin{proof}
  From the definition of convex dual, we have
  \begin{align}
    \gs(\lambda) = \max_{x} (\ip{x,\lambda} - g(x)) \le \max_{x} \big(\ip{x,\lambda} - \sum_i g_i(x_i)\big) 
                 = \sum_i \max_{x_i} (x_i \, \lambda_i - g_i(x_i)) =
                   \sum_i \gs_i(\lambda_i).\notag 
  \end{align}
  The inequality uses the superadditivity of $g$.
\end{proof}


\fi

\subsection{Supermodular Functions}
\label{sec:superm-funct}

While supermodular functions defined over the Boolean lattice are widely
considered, one can define supermodularity for all real-valued functions. Omitted proofs are presented in Appendix \ref{app:supermod}

\begin{definition}[Supermodular]
  \label{def:supermod}
  Let $X \sse \R^d$ be a lattice. A function $f: X \to \R$ is \emph{supermodular} if
  for all $x, y \in X$,
  $f(x) + f(y) \leq f(x \land y) + f(x \lor y),$
  where $x \land y$ and $x \lor y$ are the component-wise minimum and
  maximum operations.
\end{definition}
This corresponds to the usual definition of (discrete) supermodularity
when $X = \{0,1\}^d$.  For proof of the lemma below and other equivalent
definitions, see, e.g.,~\cite{Top98}.
%
\begin{lemma}[Supermodularity and Gradients]
\label{lemma:supermod}
A convex function $f : \R^d_+ \to \R$ is supermodular if and only if any
of the following are true.
\begin{OneLiners}
\item $\nabla f$ is increasing in each
  coordinate, 
  if $f$ is differentiable.
\item $\smash{\frac{\partial^2 f(x)}{\partial x_i \partial x_j}} \geq 0$
  for all $i, j$, if $f$ is twice-differentiable.
\end{OneLiners}
\end{lemma}

\ifsubmit
\else
\begin{proof}
  \emph{(Supermodular $\equiv$ inc. gradient)} Suppose $f$
  differentiable. For $x'_i \ge x_i$, consider the difference
  $f(x_{-i},x'_i) - f(x_{-i},x_i)$, where $x_{-i}$ denotes a vector of
  values for all coordinates but the $i$th one. Theorem 2.6.1 and
  Corollary 2.6.1 of~\cite{Top98} show that $f$ is supermodular iff for
  all $x'_i \le x_i$ this difference is increasing in each coordinate of
  $x_{-i}$. But this happens iff the partial derivative
  $\frac{\partial f}{\partial x_i}(x_{-i},x_i)$ is increasing in
  $x_{-i}$ (the ``if'' part follows by writing the difference
  $f(x_{-i},x'_i) - f(x_{-i},x_i)$ as an integral over the gradients,
  and the ``only if'' part follows by taking the limit of this
  difference with $x'_i \rightarrow x_i$). Moreover, since $f$ is convex
  its restriction to the line $\{(x_{-i},x_i) : x_i \ge 0\}$ is also
  convex, and hence its gradient, which is
  $\frac{\partial f}{\partial x_i}(x_{-i},x_i)$, is increasing in
  $x_i$. Thus, each partial derivative
  $\frac{\partial f}{\partial x_i}(x)$ is increasing in all coordinates
  of $x$.
	
  \medskip \emph{(Inc. gradient $\equiv$ non-neg. second derivatives)}
  Suppose $f$ is twice-differentiable. The $i$th coordinate of the
  gradient $\nabla f(x)_i$ is increasing in all coordinates of $x$ iff
  its partial derivatives
  $\frac{\partial \nabla f(x)_i}{\partial x_j} = \frac{\partial^2
    f}{\partial x_i \partial x_j}$ are non-negative for all $j$. The
  equivalence then follows.
\end{proof}
\fi




\begin{lemma}[Superadditivity]
  \label{lem:superadd2}
  If $f: \R_+^d \to \R$ is differentiable, convex, and supermodular, then 
  for $x, x', y \in \R_+^d$ such that $x' \leq x$,
  $f(x' + y) - f(x') \leq f(x + y) - f(x)$.
  In particular, if $f(0) = 0$, setting $x' = 0$ gives $f(x) + f(y) \leq f(x+y).$
\end{lemma}

\begin{corollary}[Subadditivity of profit] 
	The profit function $\profit$ is subadditive.
\end{corollary}

The next fact shows that the cost $g$ is also supermodular when seen in a discrete way. 

\begin{fact}[Continuous vs.\ Discrete Supermodularity]
  \label{fct:supermod}
  Given a convex supermodular function $g: \R^d \to \R$ and $n$ items
  with sizes $s_1, \ldots, s_n \in \R_+^d$, define the function $h:
  \{0,1\}^n \to \R$ as $h(v) = g(\sum_i s_iv_i) = g(Sv)$. Then $h(\cdot)$ is a
  (discrete) supermodular function.
\end{fact}


\ifsubmit
\else

\subsection{Probabilistic inequalities}

\begin{fact}\label{fact:expectation}
  Consider a vector $x \in \{0,1\}^n$ and let $\mathbf{X}$ be the random
  vector obtained by setting each coordinate of $x$ to $0$ with
  probability $1/2$. If cost function $g$ is supermodular then
  \begin{align*}
    \E[\profit(\mathbf{X})] \ge \frac{1}{2}\, \profit(x).
  \end{align*}
\end{fact}	

\begin{proof}
  Function $g$ is superadditive due to Lemma \ref{lem:superadd2}, and so
  $\pi$ is subadditive: $\profit(y + z) \le \profit(y) + \profit(z)$.
  %
  Writing $x = \mathbf{X} + (x-\mathbf{X})$ and applying subadditivity, we get 
  \begin{align*}
    \profit(x) \le \profit(\mathbf{X}) + \profit(x - \mathbf{X}).
  \end{align*}
  But $\mathbf{X}$ and $x - \mathbf{X}$ have the same distribution, so
  taking expectations gives
  $2 \E[\profit(\mathbf{X})] \ge \profit(x)$.
\end{proof}	

\begin{fact} \label{fact:submod-conc} Consider a
  submodular function $f : 2^\mathcal{U} \rightarrow \R$. Consider a set
  $Y \subseteq \mathcal{U}$ such that $f$ is non-negative over all of its subsets and we also have the following Lipschitz condition for some $M$:
  \begin{gather}\label{eq:lip}
    \textrm{For all $Y' \subseteq Y$ and element $e \in Y'$,~~~} |f(Y') - f(Y' - e)| \le M.
  \end{gather}
  Let $\mathbf{Y}$ be the random subset obtained from picking each
  element from $Y$ independently with some probability (which can be different
  for each item). Then
  \begin{gather}\label{eq:conc2}
    \Pr(|f(\mathbf{Y}) - \E[f(\mathbf{Y})]| \ge t) \le \frac{2M\,
      \E[f(\mathbf{Y})]}{t^2}
  \end{gather} 
\end{fact}

\begin{proof}
  Vondr\'ak showed that $M$-Lipschitz non-negative submodular functions
  are \emph{weakly (2M,0)-self-bounding}~\cite{Von}.  By the Efron-Stein
  inequality, such functions have
  $\text{Var}(f(\mathbf{Y})) \leq 2M\,
  \E[f(\mathbf{Y})]$~\cite{concentration}. Now Chebychev's inequality gives the
  result. 
\end{proof}



\fi




\newif\ifstandalone

\ifstandalone
\include{standalone}
\fi

\section{The Offline Unconstrained Problem}
\label{sec:unconstr}


We first present an offline algorithm for supermodular functions in the
\textbf{unconstrained} case (where $\calF = \{0,1\}^n$). 
We focus on the main techniques and defer
some technicalities and all computational aspects for now.
Just for this section, we assume item sizes are
``infinitesimal''. We make the following assumptions on the cost function $g$ and the
elements.  

\begin{assumption}
  \label{asm:nice}
  We assume that cost function $g$ is non-negative, strictly
  convex, closed, and differentiable. We assume $g(0) = 0$, $g$ is
  supermodular, and that gradients of $g$ go to $\infty$ along every
  positive direction.  We assume elements are in general
  position\footnote{There are no non-trivial linear dependencies, see
  Lemma~\ref{lem:reduce-nice} for a formal definition}, and that there are no exceptional 
  items. We also assume that every individual
item has profit at most $M := \nicefrac{\OPT}{\eta d}$ for
$\eta \geq 10^4$. (See Section~\ref{sec:except} on how to remove these assumptions on elements.)
\end{assumption}

\paragraph{Classifiers.} The offline algorithm will be based on \emph{linear classifiers}, where a set of weights is used to aggregates the multidimensional size of an item into a scalar, and the algorithm picks all items that have high-enough value/aggregated-size ratio. 

\begin{definition}[Classifiers and Occupancy]
  \label{def:occ}
  Given a vector $\lambda \in \R_+^d$ (a ``classifier''), define the
  set of items picked by $\lambda$ as $U_\lambda := \{ e \in U \mid v(e) \ge \ip{\lambda,s(e)} \}$.
  Let $\occ{\lambda} := \sum_{e: v(e) \geq \ip{ \lambda, s(e) }} s(e)$ denote the multidimensional occupancy induced by choosing items in $U_\lambda$.
\end{definition}


To understand the importance of classifier-based solutions it is instructive to consider the problem with single-dimensional
size. A little thought shows that an optimal solution is to pick items in decreasing order of
value density $v(e)/s(e)$. Adding these items causes the total
occupancy---and hence the incurred cost---to increase, so we stop when
the value density of the current item becomes smaller than the derivative of the cost function at the current utilization. That is, we find a density threshold $\lambda$ such that
$g'(\text{\emph{total size of items having $v(e) \ge \lambda \, s(e)$}}) \approx \lambda$, and take all these high-density items. Thus, the optimal solution is one based on the classifier $\lambda$. 

	To see that this holds in the multi-dimensional case, express $g$ in terms of linearizations
	\begin{align}
		g(z) = \max_{\lambda \in \R^d_+} (\ip{\lambda,z} - \gs(\lambda)), \label{eq:dual}
	\end{align}
	where $\gs$ is its Fenchel dual. (Note we are maximizing over positive
classifiers; Lemma~\ref{lem:pos-lambda} shows this is WLOG.) Then our unconstrained problem~(\ref{eq:1})
becomes a minimax problem:
\begin{align*}
  \max_{x \in \{0,1\}^n} \min_{\lambda \in \R^d_+} \bigg[\ip{v,x} -
  \bigg(\ip{\lambda,Sx} - \gs(\lambda)  \bigg)   \bigg]. 
\end{align*}

Consider an optimal pair $(x^*, \lambda^*)$; i.e., a pair that
is a saddle-point solution, so neither $x^*$ nor $\lambda^*$ can be
improved keeping the other one fixed. This saddle-point optimality implies:
\begin{enumerate}
\item[(a)] Since
  $\lambda^* = \argmax_{\lambda \in \R^d_+} (\ip{\lambda, Sx^*} -
  \gs(\lambda))$, it is the right linearization of $g$ at $Sx^*$ and thus $\lambda^* = \gr g(Sx^*)$ (see Claim~\ref{clm:doubleDual}).

\item[(b)] $x^*$ is such that $x^*_i = 1$ if $v_i > \ip{\lambda^*, S^i}$
  and $x^*_i = 0$ if $v_i < \ip{\lambda^*,S^i}$, with $S^i$ being the
  $i^{th}$ column of $S$ and the size of the $i^{th}$ item. \end{enumerate}

From part (b) we see that the optimal solution $x^*$ is essentially the one picked by the classifier $\lambda^*$ (ignoring coordinates with the ``0 marginal value'' $v_i = \ip{\lambda^*,S^i}$). Moreover, the converse also holds. 

\begin{claim} \label{claim:optimality}
  For a classifier $\lambda \in \R^d_+$, let $x$ be the items
  picked by it. If we have $\lambda = \nabla g(Sx) \stackrel{\text{~def~}}{=} \nabla g(\occ{\lambda})$, then $x$ is an optimal solution.
\end{claim}

\begin{proof}
  For any solution $x'$,
  \begin{align*}
    \profit(x') &= \ip{v,x'} - g(Sx') \leq \ip{v,x'} -
                  \ip{\lambda, Sx'} + \gs(\lambda) \\
                &\leq \ip{v,x} - \ip{\lambda, Sx} +
                  \gs(\lambda) \stackrel{(\lambda = \gr g(Sx))}{=}
                  \ip{v,x} - g(Sx) = \profit(x),
  \end{align*}
  where the second inequality holds since, by definition, $x$ maximizes $\ip{v,x} - \ip{\lambda, Sx}$. 
\end{proof}

\paragraph{Restricting the Set of Classifiers.} The existence
of such good classifiers is not enough, since we need to \emph{find} them online. This is difficult not only because of $d$ degrees of freedom and no control over the magnitude of the values/sizes (to be exploited in concentration inequalities), but also because picking too few or too many items could lead to low
profits. 

So we restrict the set of candidate classifiers to be a
\emph{monotone}\footnote{A curve $\calC$ is \emph{monotone} if for every
  pair $\lambda, \lambda' \in \calC$, one is coordinate-wise smaller
  than the other.}
\emph{1-dimensional} curve
$\mathcal{C} \subseteq \R^d_+$, satisfying additional
properties given below. The main motivation is that it imposes a total ordering on the set of items picked by the classifiers: given $\lambda \le \mu$ on such a curve $\calC$, the sets of items picked satisfy the inclusion 
$U_{\lambda} \supseteq U_{\mu}$.
This allows us to select a ``minimally good'' classifier in $\calC$ in a
robust way, avoiding classifiers that select too many items.

To design the curve $\mathcal{C}$ so it contains a classifier with profit $\approx \frac{\OPT}{d}$, we relax the condition $\gr g(\occ{\lambda}) = \lambda$ from Claim~\ref{claim:optimality} (too much to ask) and require the existence of $\lambda \in \calC$ satisfying:

%
%
%
\begin{enumerate}
\item[(P1)] (don't pick too many items) $\nabla g(\occ{\lambda})
    \le \lambda$.
		
\item[(P2)] (partial gradient equality) There is a coordinate $i^*$ where $(\gr g(\occ{\lambda}))_{i^*} = 
    \lambda_{i^*}$.

		
\item[(P3)] (balanced curve) 
$\gs_{i}(\lambda_{i}) = \gs_{j}(\lambda_{j}) ~~~\forall i, j \in [d]$ (see also Claim~\ref{clm:dual-marginal}).

\end{enumerate}

Property (P1) enforces half of the equality in Claim~\ref{claim:optimality}, and (P2) guarantees that
equality holds for \emph{some} coordinate.
Now for property (P3). Since 
$\lambda \neq \gr g(\occ{\lambda})$ the optimality proof of Claim \ref{claim:optimality}
does not go though, since
$g(\occ{\lambda}) \neq \ip{\lambda,\occ{\lambda}} - \gs(\lambda)$. As we prove later, the difference between these terms can be at most $\gs(\lambda)$ (see
Figure \ref{fig:offshoot} for an illustration), and the  superadditivity
of $g$ gives us $\gs(\lambda) \leq \sum_i \gs_i(\lambda_i)$ (see Claim~\ref{claim:subaddCoord}). Property~(P3) is used to
control this sum, by charging it to the coordinate $i^*$ where we know
we have ``the right linearization'' (by property (P2)). 
	Reinterpreting the construction of~\cite{BUCM} in our setting, we then define $\calC$ as any monotone curve where every $\lambda \in \calC$ satisfies (P3).
	
	\begin{lemma} \label{lemma:existC}
		The curve $\calC$ exists and contains a $\lambda$ satisfying properties (P1)-(P3).
	\end{lemma}
	
	\begin{proof}
		We first show existence, that is, the set $\{ \lambda \in \R^d_+ \mid \gs_i(\lambda_i) = \gs_j(\lambda_j) ~~\forall i, j \}$ contains a monotone curve. Notice that this set is the union of the box $\{\lambda \in \R^d_+ \mid \gs_i(\lambda_i) = 0~~\forall i\} = \prod_i [0, g'_i(0)]$ (range of slopes where we can swivel around $g_i(0) = 0$) and a monotone curve 
$\{\lambda(\tau) \mid \tau > 0\}$, where $\lambda(\tau)$ is
the unique vector satisfying $\gs_i(\lambda_i(\tau)) = \tau$; uniqueness follows from the fact $\gs_i$ stays at value zero in the interval $[0, g'_i(0)]$, but after that is strictly increasing due to its convexity, and monotonicity of this curve also follows from monotonicity of the $\gs_i$'s. Thus, $\CC$ is this curve plus any monotone curve extending it to the origin.
%


	To see that $\calC$ satisfies properties (P1) and (P2), we note that since the $\gs_i$'s are increasing and not identically 0, 
$\calC$ is unbounded in all
coordinates. Thus, a sufficiently large $\lambda \in \calC$ satisfies (P1), and we can start with such $\lambda$ and move down the curve (decreasing in each coordinate) until we obtain $\lambda' \in \calC$ with $\lambda' = \nabla g(\occ{\lambda'})$, since the 
$g$ has increasing gradients. (The equality in this final step uses the assumption that item sizes are infinitesimal, which we made for simplicity in this section).
	\end{proof}

	Making the above discussion formal, we show that $\calC$ has a high-value classifier.
	Recall that $U_\lambda$ is the set of items picked by $\lambda$ (Definition \ref{def:occ}).

\begin{theorem}
  \label{thm:unc}
  Given Assumption~\ref{asm:nice}, let $\lambda^*$ be a classifier in
  $\calC$ satisfying properties~(P1)-(P3). Then for all $x' \in
    [0,1]^n$ we have $\pi(U_{\lambda^*}) \ge \frac{1}{d+1}\cdot \pi(x')$.
\end{theorem}

\begin{proof}
  Let $x^* = \chi_{U_\lambda^*}$ be the solution picked by the
  classifier $\lambda^*$, and note that $\occ{\lambda^*} = Sx^*$.  Let
  $L(y,\mu) := \ip{v,y} - [\ip{\mu, Sy} - \gs(\mu)]$ be the
  linearization of $\pi(y)$ at some slope $\mu$. From \eqref{eq:dual} we know $g(y) \ge L(y,\mu)$ for all $\mu \ge 0$. Since $x^* $ is optimal for the linearization $L(y,\lambda^*)$ (because $x^*_i = 1$ iff
  $v_i - \ip{\lambda^*, S^i} \ge 0$), we have 
  \begin{align}
    L(x^*, \lambda^*) \ge L(x', \lambda^*) \ge \profit(x')~~~~~~~\textrm{for all
    $x' \in [0,1]^n$.} \label{eq:thmU1}
  \end{align}
  Now we relate the true profit $\profit(x^*)$ to this linearized
  value. Observe that
  \begin{align}
    \profit(x^*) = \ip{v,x^*} - g(Sx^*)
    &= \ip{v,x^*} - [\ip{\nabla
      g(Sx^*), Sx^*} - \gs(\nabla g(Sx^*))] \tag{by
      Claim~\ref{clm:linear}} \\
    & \geq \underbrace{\ip{v,x^*} - \ip{\lambda^*, Sx^*}}_{\geq 0} +
      \underbrace{\gs(\nabla g(Sx^*))}_{\geq 0}, \label{eq:7}
  \end{align}
  where the inequality uses that $\lambda^* \ge \nabla g(Sx^*)$ by
  property~(P1) and $Sx^* \ge 0$. The first term is
  non-negative because we only pick items for which
  $v_i - \ip{\lambda, S^i} \ge 0$. The second term is non-negative due
  to Claim~\ref{clm:fenchel-props}(a). We can now prove three lemmas that imply the theorem.

  \begin{lemma}
    \label{lemma:linGap}
    For any $x' \in [0,1]^n$,
    $\profit(x^*) \ge L(x^*, \lambda^*) - \gs(\lambda^*) \geq
    \profit(x') - \gs(\lambda^*).$
  \end{lemma}

\begin{subproof}
  Drop the second term from~(\ref{eq:7}), then use the definition of
  $L(\cdot,\cdot)$ and~(\ref{eq:thmU1}).
\end{subproof}

\begin{lemma}\label{lemma:gsi}
  $\gs(\lambda^*) \le d \cdot \gs_{i^*}(\lambda^*_{i^*}).$
\end{lemma}
\begin{subproof}
  Using the superadditivity of $g$ and Claim~\ref{claim:subaddCoord} we
  get $\gs(\lambda^*) \le \sum_i \gs_i(\lambda^*_i)$. Now from property (P3) of the classifier $\lambda^*$, all the terms in the sum
  are equal.
\end{subproof}

\begin{lemma} \label{lemma:profitDual}
  $\profit(x^*) \ge \gs_{i^*}(\lambda^*_{i^*}).$
\end{lemma}

\begin{subproof}
  We claim that $\gs(\nabla g(Sx^*)) \ge \gs_{i^*}(\lambda^*_{i^*})$;
  plugging this into~(\ref{eq:7}) proves the lemma.  For the claim,
  define $\lambda' = \nabla g(Sx^*)$.  By Property~(P2),
  $\lambda'_{i^*} = \lambda^*_{i^*}$, so we want to show
  $\gs(\lambda') \geq \gs_{i^*}(\lambda'_{i^*}) = \gs(\lambda'_{i^*}
  \be_{i^*})$.  This follows because $\gs$ is monotone
  (Claim~\ref{clm:fenchel-props}(b)).
\end{subproof}
This completes the proof of Theorem~\ref{thm:unc}.
\end{proof}

\ifstandalone
\end{document}
\fi



\newif\ifstandalone

\ifstandalone
\include{standalone}
\fi

\newcommand{\amax}{\textrm{argmax}}
\newcommand{\supp}{\textrm{support}}

\section{The Offline Constrained Case}
\label{sec:constr}

Having built up tools and intuition in the unconstrained case, we turn
to the case where there is a downwards-closed constraint
$\F \subseteq \{0,1\}^n$, and the goal is to maximize the profit subject
to $x \in \F$. We again work with Assumption~\ref{asm:nice}, but do not
assume anything about items sizes. We discuss computational aspects
at the end of this section.

The general idea is again to use classifiers $\lambda \in \R^d_+$, and only consider items in
$U_{\lambda}$, namely those with ``high-enough'' value
$v_i \ge \ip{\lambda, S^i}$. However, because of the constraints $\F$ we
may no longer be able to pick all these items. Thus, we need to consider the
most profitable solution from $\calF$ in this filtered feasible set
$U_{\lambda}$ (whose quality is less clear how to analyze).



	Again we restrict to the 1-dimensional curve $\calC$ defined in the previous section; however, it only satisfies slightly modified versions of properties (P1)-(P2),
since we do not assume the item sizes to be infinitesimal anymore. To make this precise, define the ``open'' set
$\Aopen_\lambda := \{ e \in U \mid v(e) > \ip{ \lambda, s(e) } \}$; note
the strict inequality. Under the assumption of items being in general
position, there is at most one ``threshold'' item with
$v_i = \ip{\lambda, S^i}$, i.e.,
$|U_{\lambda} \setminus \Aopen_\lambda| \leq 1$.  Now a ``good''
classifier is one that satisfies the following:
\begin{enumerate}
\item[(P1')] For all binary $x$ with $\supp(x) \sse \Aopen_\lambda$ and $x \in \F$, $\gr g(Sx) \leq \lambda$.

\item[(P2')] There exists a binary $\xocc$ with $\supp(\xocc) \sse U_{\lambda}$
  and $\xocc \in \calF$, and index $i^*$ such that $(\gr g(S\xocc))_{i^*} \geq \lambda_{i^*}.$
  (Note that if $\supp(\xocc) \sse \Aopen_\lambda$, then by property~(P1') the above inequality holds at equality; else $\xocc$ contains the unique element in $U_{\lambda} \setminus \Aopen_{\lambda}$.)
\item[(P3')] This is the same as before: $\gs_{i}(\lambda_{i}) = \gs_{j}(\lambda_{j}) ~~~\forall i, j \in
    [d]$.
\end{enumerate}

\newcommand{\xlin}{\ensuremath{x^{lin}}}

	The arguments of Lemma~\ref{lemma:existC} show the following. 
	
	\begin{lemma}
		Given Assumption \ref{asm:nice}, the curve $\calC$ defined in the previous section contains a $\lambda$ satisfying properties (P1')-(P3').
	\end{lemma}
Next, we show that for a good classifier $\lambda \in \calC$, the maximum profit solution from $\calF$ contained
within $\Aopen_{\lambda}$ essentially gives an $O(\nicefrac{1}{d})$-approximation.

\begin{theorem}[Offline Approach]
  \label{thm:constr-main}
  Suppose Assumption \ref{asm:nice} holds. Let
  $\lambda^*$ be a classifier in $\calC$ satisfying
  properties~(P1')--(P3'). Then the better of
  the two solutions: (a)~the maximum profit solution in $\calF$
  containing elements only from $\Aopen_{\lambda^*}$, and (b)~the optimal
  single element in $U_{\lambda^*}$, has profit at least $\pi(x')/(2d+1)$
  for any vector $x' \in \conv(\F) \sse [0,1]^n$.
\end{theorem}

\begin{proof}

	The idea is to follow the development in Theorem \ref{thm:unc}. There same solution $x^*$ satisfied the value lower bounds of Lemmas \ref{lemma:linGap} and \ref{lemma:profitDual}; to satisfy the first lemma, we needed the solution to be optimal for the linearization of $\pi$ using ``slope'' $\lambda^*$; to satisfy the second, we needed to satisfy (P2). Here, we construct two solutions in $\F$ intersect $U_{\lambda^*}$ to satisfy these lemmas separately: 
%
\ignore{To motivate these solutions, we go back to the analysis for the unconstrained case in Theorem~\ref{thm:unc}. There, we showed that
  solution $x^{unc} = \chi_{U_{\lambda^*}}$ had large profit using the
  following properties:

  \begin{enumerate}
  \item[(i)] $x^{unc}$ was the \emph{maximum profit solution for the
      linearized objective}
    $L(\cdot, \lambda^*) = \ip{v - S^\intercal \lambda^*,\cdot} +
    \gs(\lambda^*)$, where the cost $g$ was linearized with the good
    classifier $\lambda^*$. This allowed us to use the dual $\lambda^*$ to
    lower bound the profit of the solution up to the linearization gap
    (Lemma~\ref{lemma:linGap}).
	
  \item[(ii)] $x^{unc}$ had \emph{large occupation} in the coordinate
    $i^*$: that is, $\nabla g(Sx^{unc})_{i^*} = \lambda^*_{i^*}$. This
    allowed us to pay for the linearization gap
    (Lemma~\ref{lemma:profitDual}).
  \end{enumerate}

  Without constraints, both optimality and occupation were achieved
  simultaneously by picking all items in $U_{\lambda^*}$. However, having
  constraints $\F$ means we now use a separate solution for
  each.}
  \begin{alignat*}{2}
    \xlin &:= \amax\{ \ip{v,y}
    - \ip{\lambda^*,Sy} \mid y \sse \Aopen_{\lambda^*}, y \in \F\}\\
    \xocc &:= \textrm{the solution
      promised by property~(P2')}.
  \end{alignat*}
  
  Since property~(P1') and~(P3') holds for
  $\xlin$, Lemmas~\ref{lemma:linGap} and~\ref{lemma:gsi} hold
  essentially unchanged, and thus for any vector $x' \in \conv(\calF)$ we have
  \begin{align}
    \profit(\xlin) \ge \profit(x') - d \cdot \gs_{i^*}(\lambda^*_{i^*}). \label{eq:4}
  \end{align}

	The solution $\xocc$ may not belong to the set $\Aopen_{\lambda^*}$, since it may contain the threshold item $e^{\circ} = \ip{\lambda^*, s(e^{\circ})}$, if it exists (let $x^\circ = \chi_{\{e^\circ\}}$ be its characteristic vector, all 0's vector if does not exists). Let $\xrest = \xocc - x^\circ$.
	
  \begin{lemma}
    \label{lem:rest}
		These solutions satisfy $\profit(\xrest) + \profit(x^{\circ}) \geq
      \gs_{i^*}(\lambda^*_{i^*}).$
  \end{lemma}
  
  \begin{subproof}
    Property~(P1') gives $\gr g(S\xrest) \leq \lambda^*$, and
    Property~(P2') implies
    $\gr g(S(\xrest + x^{\circ})) = \gr g(S\xocc)$ is at least $\lambda^*$
    at some coordinate $i^*$. Since $g$ is convex and differentiable,
    the gradients are continuous~\cite[Remark~D.6.2.6]{HUL}, so there is $\delta \in [0,1]$ where the vector
    $\widehat{x} := \xrest + \delta x^\circ$ satisfies
    $\gr g(S\widehat{x}) \le \lambda^*$ and $\gr g(S\widehat{x})_{i^*} = \lambda^*_{i^*}$ for some coordinate $i^*$.
    Due to these properties, the proof of
    Lemma~\ref{lemma:profitDual} holds for 
    $\widehat{x}$ and shows
    $\profit(\widehat{x}) \ge \gs_{i^*}(\lambda^*_{i^*})$.

    The assumption of no exceptional items gives
    $\profit(\delta x^\circ) \leq \profit(x^\circ)$. From 
    subadditivity of profit~$\profit$,
    %
    $\gs_{i^*}(\lambda^*_{i^*}) \le \profit(\widehat{x}) \le \pi(x^{rest})
                               + \pi(\delta x^{\circ}) \le \pi(x^{rest}) + \pi(x^{\circ}).$
    This concludes the proof.    
  \end{subproof}
  Combining Lemma~\ref{lem:rest} with inequality~\eqref{eq:4}, for any
  $x' \in \F$ we have
  \begin{gather*}
    \pi(x') \leq \profit(\xlin) + d\,\profit(\xrest) + d\,\profit(x^\circ).
  \end{gather*}
  Since $\xlin, \xrest$ are feasible for problem~(a) in the theorem
  statement, and $x^\circ$ is feasible for problem~(b), the best of them gives a $(2d+1)$-approximation. This proves
  Theorem~\ref{thm:constr-main}.
\end{proof}



Picking the most profitable singleton is trivial offline, and
well-approximable online by the secretary
algorithm~\cite{freeman-secretary}. Moreover, we need to approximately optimize the \emph{submodular} function $\profit$ (Fact~\ref{fct:supermod}) over $\F|_{\Aopen_{\lambda^*}}$ (i.e., the sets in $\F$ with only elements of $\Aopen_{\lambda^*}$). For several constraint structures (e.g., matroids, $p$-systems), there are known algorithms for approximately optimizing \emph{non-negative} (and sometimes also monotone) submodular functions.  Unfortunately, our profit function $\profit$ may take negative values, so we cannot directly use these algorithms. Simply considering the truncated function $\max\{\pi(z), 0\}$ does not work because it may be non-submodular. In the next section, when $g$ is \emph{separable}, we introduce a way of making our profit function non-negative everywhere, while maintaining submodularity and preserving the values at the region of interest $\F|_{\Aopen_{\lambda^*}}$.


\subsection{Making the Profit Function $\profit$ Non-negative}
\label{sec:monoton}

We first show that $\profit$ already satisfies the desired properties over the sets in $\F|_{\Aopen_{\lambda^*}}$.

\begin{lemma} \label{lemma:mono-feas}
  The profit function $\profit$ is non-negative monotone over $\F|_{\Aopen_{\lambda^*}}$.
\end{lemma}

\begin{proof}
  Since $\profit(\emptyset) = 0$ it suffices to show
  monotonicity. Consider $x \in \F|_{\Aopen_{\lambda^*}}$
  and let $\chi_e$ be the indicator os an item in $x$. Comparing the costs with and without $e$ we have
  \begin{align*}
    g(Sx) \stackrel{\text{(convexity)}}{\leq} g(S(x-\chi_e)) + \ip{ \gr g(Sx),
    S\chi_e } \stackrel{\text{(Property~(P1'))}}{\leq} 
    g(S(x - \chi_e)) + \ip{ \lambda^*, s(e)}. 
  \end{align*}
  Since $x \in \Aopen_{\lambda^*}$, we have $v(e) > \ip{\lambda^*, s(e)}$ and thus $\profit(x) > \profit(x - \chi_e)$, i.e., monotonicity.
\end{proof}

However, to run algorithms that approximately optimize $\profit$ over
$\F|_{\Aopen_{\lambda^*}}$ in a black-box fashion, non-negativity over the feasible sets $\F|_{\Aopen_{\lambda^*}}$ is not enough, even if the
algorithm only probes $\profit$ over these sets, since their \emph{proof of correctness} may require this property outside of feasible sets. Thus, we need to modify $\profit$ to ensure non-negativity outside of $\F|_{\Aopen_{\lambda^*}}$.

For that, the idea is to truncate the gradient of the cost $g$ so $\nabla g(Sx)$ becomes at most $\lambda^*$ for all subsets $x \subseteq \Aopen_{\lambda^*}$ (i.e., so Property~(P1') holds for all subsets); this was the crucial element for the monotonicity (and hence non-negativity) proof above. Notice that since Property~(P1') guarantees already $\nabla g(Sx) \le \lambda^*$ for all $x \in \F|_{\Aopen_{\lambda^*}}$, this does not change the value of $\profit$ over these points. The proof of the lemma is given in Appendix \ref{app:proofs}.



\begin{lemma}\label{lemma:piplus}
	If $g$ is separable, there is a submodular function $\profit^+$ satisfying the following:
  \begin{OneLiners}
	  \item[i.] $\profit^+$ is non-negative and monotone over all subsets of $\Aopen_{\lambda^*}$, and
  \item[ii.] $\profit^+(x) = \profit(x)$ for every $x \in \F|_{\Aopen_{\lambda^*}}$.
 	\end{OneLiners}
\end{lemma}


\subsection{The Offline Algorithm: Wrap-up}
\label{sec:off-endgame}

Using this non-negativization procedure, we get an $O(d)$-approximation
\emph{offline} algorithm for constrained profit maximization for
\emph{separable} cost functions $g$; this is an offline analog of
Theorem~\ref{thm:main2}. For the unconstrained case,
Lemma~\ref{lemma:mono-feas} implies that the profit function $\profit$
it itself monotone, so we get an $O(d)$-approximation offline algorithm
for the \emph{supermodular} case. In the next section we show how to
convert these algorithms into online algorithms.

One issue we have not discussed is the computational cost of
finding $\lambda^*$ satisfying (P1')--(P3').
In the full version of the paper, 
we show that for any $\e > 0$ we
can efficiently find a $\lambda^*$ satisfying (P1'),
(P2'), and a slightly weaker 
condition: $| \gs_{i}(\lambda^*_{i}) - \gs_{j}(\lambda^*_{j})| \leq 2\e$ for all $i, j \in
  [d]$.
Using this condition in Theorem~\ref{thm:constr-main} means we get a
profit of at least
$\frac{\Opt - 2d\e}{2d+1} \geq [\nicefrac{\Opt}{(2d+1)}] - \eps$; the
running time depends on $\log \eps^{-1}$ so we can make this loss
negligible. 


\ifstandalone
\end{document}
\fi



\newif\ifstandalone

\ifstandalone
\include{standalone}
\fi

\section{The Online Algorithm} \label{sec:contr-online}

\renewcommand{\Isample}{L}
\renewcommand{\Ioos}{R}
\newcommand{\IoosF}{R^\circ}
\renewcommand{\lhat}{\mu}
\renewcommand{\UB}{M_i^+}

In the previous sections we were working offline: in particular, in
computing the ``good'' classifier $\lambda \in \calC$, we assumed
knowledge of the entire element set. We now present the online framework
for the setting where elements come in random order. Recall the
definition of the curve $\CC$ from \S\ref{sec:unconstr}, and the fact
that there is a total order among all $\lambda \in \CC$. Recall that for
simplicity we restrict the constraints $\calF$ to be matroid constraints.

%


%

For a subset of elements $A\subseteq U$, let $\OPT(A)$
and $\FOPT(A)$ denote the integer and fractional optimal profit for $\F|_A$, the feasible solutions 
restricted to elements in $A$. Note that in the fractional case this means the best solution in the convex hull $\conv(\F|_A)$. Clearly, $\FOPT(A) \ge
\OPT(A)$. We use $\OPT$ and $\FOPT$ to denote $\OPT(U)$ and $\FOPT(U)$
for the entire instance~$U$.

Again we work under
Assumption~\ref{asm:nice}. We will also make use of any algorithm for maximizing submodular functions over $\F$ in the random-order model satisfying the following.

\begin{assumption} \label{ass:submod-algo} Algorithm $\submod$ takes a
  nonnegative monotone submodular function $f$ with $f(\emptyset) = 0$,
  and a number $N$. When run on a sequence $X$ of $N$ elements presented in random order, it returns a (random) subset
  $X_{alg} \in \calF$ with expected value
  $\E[f(X_{alg})] \geq \frac{1}{\alpha} \max_{X' \in \calF} f(X)$.
  Moreover, the it only evaluates the function $f$ on
  \textbf{feasible} sets.
\end{assumption}
Our algorithm is very simple:

\begin{algorithm}[H]
  \caption{Online Algorithm for Profit Maximization}
  \begin{algorithmic}[1]
  	\small
    \State $\Isample \gets$ first $\text{Binomial}(n, \nicefrac12)$
    items.

    \State $\lhat \gets $ largest vector on curve $\mathcal{C}$
    s.t.\    $\FOPT(\Isample_{\lhat}) \ge \frac{1}{12 d}\,
    \FOPT(\Isample)$.\label{line:2}

		\State $\Ioos \gets $ remaining instance, namely the last
    $n - |L|$ items.

    \State
    $\IoosF_{\lhat} \gets \{ e \in \Ioos \mid v(e) > \ip{ \lhat, s(e)}
    \}$ be the (strictly) ``filtered'' remaining instance.
  
    \State {\bf Un-constrained:} Select items in $\IoosF_{\lhat}$
     not decreasing the current value of the solution.

    \hspace{-0.7cm}{\bf Constrained:} Run algorithm $\submod$ on
    $\IoosF_{\lhat}$ using the profit function $\profit$, selecting items
    according to this algorithm. However, do not add any items that
    decrease the current value of the solution.
  \end{algorithmic}
  \label{alg:online-constr}
\end{algorithm}	

Note that $\Isample_{\lhat}$ denotes the set of items in the sample $L$ picked by $\lhat$ (Definition \ref{def:occ}).
In Step~\ref{line:2}, we can use the Ellipsoid method to find
$\FOPT$, i.e., to maximize the concave profit function $\profit$ over the matroid polytopes $\conv(\calF|_{\Isample_{\lhat}}$) and $\conv(\calF|_{\Isample}$), within negligible error. Moreover, we must do this
for several sets $\Isample_{\lhat}$ and pick the largest one on $\calC$ using a binary-search procedure. We defer the technical details to the full version of the paper.


\subsection{Analysis}

To analyze the algorithm, we need to show that the classifier $\lhat$ learned in Step~\ref{line:2} is large enough that we do not waste
space with useless items, but low enough that we admit enough useful
items. Along the way we frequently use the concentration bound from
Fact~\ref{fact:submod-conc}. For this we need the profit function
$\profit$ to satisfy a Lipschitz-type condition (\ref{eq:lip}) on the
optimal solutions of any given sub-instance. To facilitate this, let us
record a useful lemma, proved in Appendix \ref{app:proofs}. For a vector $y \in \R^n$, and a subset
$A \sse U$, define $y_A$ to be the same as $y$ on $A$, and zero outside
$A$.

\begin{claim}
  \label{clm:lip}
  Consider any $U' \subseteq U$, and let $y$ be an     optimal fractional solution on $\F|_{U'}$ (so $\profit(y) = \FOPT(U')$). Then for
  any $B \sse A \sse U'$ with $|A \setminus B| = 1$, we have 
  $| \profit(y_A) - \profit(y_B) | \leq M$, where
  $M$ is an upper bound on the profit from any
  single item.
\end{claim}

From Section~\ref{sec:constr}, recall $\lambda^* \in \R^d_+$ is a classifier that satisfies
properties~(P1')--(P3').
\begin{lemma}[Goldilocks Lemma]\label{lem:goldilocks}
  \label{lemma:online-good-lambda} Given Assumption \ref{asm:nice}, the classifier $\lhat$ computed in
  Line~\ref{line:2} of Algorithm~\ref{alg:online-constr} satisfies:
  \begin{enumerate}
  \item[(a)] (Not too small) $\lhat \ge \lambda^*$, with probability at
    least $\nicefrac{19}{20}$.
  \item[(b)] (Not too big) $\FOPT(U_{\lhat}) \ge \frac{\FOPT}{48d}$ with
    probability at least $1-\nicefrac{1}{20d} \geq \nicefrac{19}{20}$.
  \end{enumerate}
\end{lemma}

\begin{proof}[Proof sketch]
	(See Appendix \ref{app:proofs} for full proof.)
  For the first part, we show that the classifier $\lambda^*$ satisfies the properties needed in
  Line~\ref{line:2} with probability $1 - \nicefrac{1}{20}$; since  $\lhat$ is the largest such vector, we get $\lhat \geq \lambda^*$. Using Theorem~\ref{thm:constr-main} and the assumption that no item has large profit, we have $\FOPT(U_{\lambda^*}) \ge  \frac{\FOPT}{3d}$. Moreover, the sample obtains at least half of this profit in expectation, i.e., $\E\,\FOPT(L_{\lambda^*}) \ge \frac{\FOPT}{3d}$. Then using concentration (Fact
  \ref{fact:submod-conc}) with the Lipschitz property of Claim \ref{clm:lip} and the no-high-profit-item assumption, we have $\FOPT(L_{\lambda^*}) \ge \frac{\FOPT}{12d}$ (which is at least $\frac{\FOPT(L)}{12d}$) with probability at least $\nicefrac{19}{20}$. Thus, with this probability $\lambda^*$ satisfies the properties needed in Line~\ref{line:2} of the algorithm, as desired.

  For the part (b) of the lemma, notice that for each scenario
  $\FOPT(U_{\lhat}) \ge \FOPT(\Isample_{\lhat})$, since feasible
  solutions for the sample are feasible for the whole instance. Next, by
  definition of $\lhat$,
  $\FOPT(\Isample_{\lhat}) \geq \frac{\FOPT(\Isample)}{12d}$.  Finally,
  if $x$ is the fractional optimal solution on $U$ with
  $\profit(x) = \FOPT$, then $\E[\pi(x_\Isample)] \ge \FOPT/2$, since $g$ is super-additive. Again using the concentration bound
  Fact~\ref{fact:submod-conc}, the profit $\pi(x_\Isample)$ is at least
  $\frac{\FOPT}{4}$ with probability at least
  $(1-\nicefrac{1}{20d})$. Of course, $\FOPT(L) \geq \profit(x_L)$.
  Chaining these inequalities,
  $\FOPT(U_{\lhat}) \ge \frac{\FOPT}{48d}$ with this probability.
\end{proof}

	In view of Theorem \ref{thm:constr-main}, we show the filtered out-of-sample instance
  $\IoosF_{\lhat}$ behaves like $\Aopen_{\lambda^*}$. 

\begin{lemma} \label{lem:constr-oos} The filtered out-of-sample instance
  $\IoosF_{\lhat}$ satisfies the following w.p.\ $\nicefrac{19}{20}$:
  \begin{OneLiners}
  \item[(a)] For all $e \in \IoosF_{\lhat}$,
    $v(e) \ge \langle {\lambda^*}, s(e) \rangle$.
  \item[(b)] For all $x$ with
    $\text{support}(x) \subseteq \IoosF_{\lhat}$ such that $x\in \calF$,
    $\nabla g(Sx) \le {\lambda^*}$.
  \item[(c)] $\FOPT(\IoosF_{\lhat}) \ge \frac{\FOPT}{200d}$.
  \end{OneLiners}
\end{lemma} 
\begin{proof}
  By Lemma~\ref{lem:goldilocks}(a), threshold $\lhat \ge {\lambda^*}$
  with probability $\nicefrac{19}{20}$. When that happens,
  $\Aopen_\lhat \sse \Aopen_{\lambda^*}$.  Since the first two
  properties hold for $\Aopen_{\lambda^*}$, they also hold for
  $\Aopen_{\lhat}$, and by downward-closedness, also for
  $\IoosF_{\lhat}$.

  For the third part, let $\lambda^+$ be the largest threshold in
  $\mathcal{C}$ such that $\FOPT(U_{\lambda^+}) \ge \frac{\FOPT}{48
    d}$. From Lemma~\ref{lemma:online-good-lambda}(b), with good
  probability we have $\lhat \le \lambda^+$. Since $\lhat$ is a smaller
  threshold, the instance $U_{\lambda^{+}}$ is contained in the instance
  $U_{\lhat}$, which implies that for every scenario
  $\FOPT(\Ioos_{\lhat}) \ge \FOPT(\Ioos_{\lambda^+})$. Next we will show
  that that with good probability
  $\FOPT(\Ioos_{\lambda^+}) \ge \frac{\FOPT}{200d}$, and hence get the
  same lower bound for $\FOPT(\Ioos_{\lhat})$. If $y$ is the optimal
  fractional solution for $U_{\lambda^+}$, then $y_\Ioos$ is feasible
  for $\Ioos_{\lambda^+}$ with
  $\E[\profit(y_\Ioos)] = \frac12 \FOPT(U_{\lambda^+}) \geq
  \frac{\FOPT}{96d}$. Moreover, using the concentration bound again, we
  get that $\profit(y_\Ioos) \geq \frac{\FOPT}{192d}$ with probability
  at least $\nicefrac{19}{20}$. 
  Finally, by the assumption of general position, there is at most one
  item in $\Ioos_{\lhat} \setminus \IoosF_\lhat$. Dropping this item
  from the solution $y$ to get $y^\circ$ reduces the value by at most
  $M = \frac{\FOPT}{10^4d}$; here we use subadditivty of the profit,
  and that there are no exceptional items.  Hence, with probability at
  least $\nicefrac{19}{20}$:
  \begin{gather*}
    \FOPT(\IoosF_{\lhat}) \ge \FOPT(\IoosF_{\lambda^+}) \ge \profit(y^\circ_\Ioos)
    \geq \frac{\FOPT}{196d} - M \geq \frac{\FOPT}{200d}.  \qedhere
  \end{gather*}
\end{proof}
Finally, we are ready to prove the main theorems in the online setting.

\begin{theorem}[Unconstrained Case: Supermodular Cost Functions]
  \label{thm:online-unconstr}
  Algorithm \ref{alg:online-constr} gives an $O(d)$-approximation in
  expectation for the unconstrained case, if the cost function is
  supermodular.
\end{theorem}

\begin{proof}
  Define the event $\calE$ that Lemmas~\ref{lemma:online-good-lambda}
  and~\ref{lem:constr-oos} hold; $\Pr(\calE) \geq
  \nicefrac{17}{20}$. Now, by Lemma~\ref{lem:constr-oos}(c), the optimal
  fractional solution for $\IoosF_{\lhat}$ has profit at least
  $\FOPT/200d$. Moreover, since there are no constraints, the profit
  function is monotone submodular over all of $U^{\circ}_{\lambda^*}$ by
  Lemma~\ref{lemma:mono-feas}. Conditioning on the good event $\calE$,
  Lemma~\ref{lemma:online-good-lambda}(a) gives that $\IoosF_{\lhat}
  \sse U^{\circ}_{\lambda^*}$, so the algorithm to maximize the monotone
  submodular function (both integrally and fractionally) is to pick all
  elements. Hence, conditioned on $\calE$, the profit we get is at least
  $\FOPT/200d$. In the other case, we never pick an item that gives
  negative marginal value, so our solution is always non-negative. Hence
  our expected profit is at least
  $\Pr[\calE] \cdot \OPT(\Ioos_{\lhat}) = \Omega(\FOPT/d)\ge
  \Omega(\OPT/d)$.
\end{proof}

The analysis of the algorithm for the constrained separable-cost case is similar, only using the constrained offline guarantees of Theorem \ref{thm:constr-main}, and the non-negativization Lemma \ref{lemma:mono-feas} to argue that $\submod$ maintains its guarantees. Details are provided in Appendix \ref{app:proofs}.

\begin{theorem}[Constrained Case: Separable Cost Functions]
  \label{thm:online-constr}
  Suppose algorithm $\submod$ satisfies Assumption~\ref{ass:submod-algo}
  and is $\alpha$-competitive in expectation.  Then Algorithm
  \ref{alg:online-constr} gives a $O(\alpha d^2)$-approximation in
  expectation.
\end{theorem}



\ifstandalone
\end{document}
\fi


\section{Separability versus Supermodularity}
\label{sec:separ-vs-super}

In this section, we show that an $\beta$-approximation algorithm for
the separable-cost case gives a $O(d\beta)$-approximation for a slight
generalization of the supermodular-cost case.
Consider the problem of picking a set $A$ to solve
\[ \profit(A) := \max_{A \in \F} \bigg( v(A) - g\big(\sum_{e \in A}
    s(e)\big) \bigg), \] where $v(A)$ is a (discrete) submodular
function over $\{0,1\}^n$ with $v(\emptyset) = 0$, $g$ is a convex,
(continuous) supermodular function over $\R^d$, and $\F$ is some
downward-closed constraint set. We show that for the case of matroid
constraints, this problem can be reduced to the setting where the cost
function is separable over its $d$ coordinates,  suffering a
loss of $O(d)$.

\begin{theorem}[Reduction]
  \label{thm:reduction}
  Given an $\beta$-approximation algorithm for profit-maximization for
  \emph{separable} convex cost functions under matroid constraints, we
  can get an $d(\beta + 2\mathrm{e} d)$-approximation algorithm for the
  profit-maximization problem with \emph{supermodular costs} $g$,
  \emph{submodular values} $v$, and $\F$ being a \emph{matroid} constraint.
\end{theorem}
The reduction is the following:
\begin{OneLiners}
\item[1.] Define separable costs
  $\overline{g}(y) := \nicefrac1d \sum_{i = 1}^d g_i(d y_i)$, where
  $g_i$ are marginal functions for $g$.
\item[2.] W.p.\ $p = \frac{\beta}{\beta + 2\mathrm{e} d}$, run 
  single-secretary algorithm to return element with maximum profit.
\item[3.] W.p.\ $1-p = \frac{2\mathrm{e} d}{\beta + \mathrm{e} d}$, run 
  algorithm for value function $v(\cdot)$ and separable cost fn.\
  $\overline{g}(\cdot)$.
\end{OneLiners}
This reduction relies on the following simple but perhaps
surprising observation that relates separability with supermodularity,
which may find other applications.

\begin{lemma} 
  \label{lemma:dapprox} 
  Given a monotone convex superadditive function $g$ with $g(0) = 0$,
  let $g_i$ be the marginal functions. Then for all $y \in \R^d_+$:
  \begin{enumerate}
  \item $g(y) \ge \sum_i g_i(y_i)$
  \item $g(y) \le \frac{1}{d} \sum_i g_i(d y_i) = \overline{g}(y)$. 
  \end{enumerate}
\end{lemma}
	
\begin{proof}
  The first property follows from the superadditivity of $g$, and the
  second follows from Jensen's inequality.
\end{proof}

%
%

While the full proof of Theorem \ref{thm:reduction} is deferred to
Appendix \ref{app:proofs}, the main idea is clean. Given an optimal
integer solution $x^*$ for the original problem (with the original cost
function), we use Lemma~\ref{lemma:dapprox} and the Lov\'asz (convex)
extension of submodular functions to show that $x^*/d$ is a good
fractional solution for the separable cost function. Now using
polyhedral properties of $d$-dimensional faces of the matroid polytope,
and other properties of the Lov\'asz extension, we show the existence of
a good integer solution to the separable problem.
Combining this reduction
with Theorem~\ref{thm:main2} proves Theorem~\ref{thm:main2b}.



{\small 
\bibliography{bib}
\bibliographystyle{amsalpha}
}

\appendix

\begin{figure}[htbp]
  \centering
  \includegraphics[scale=0.57]{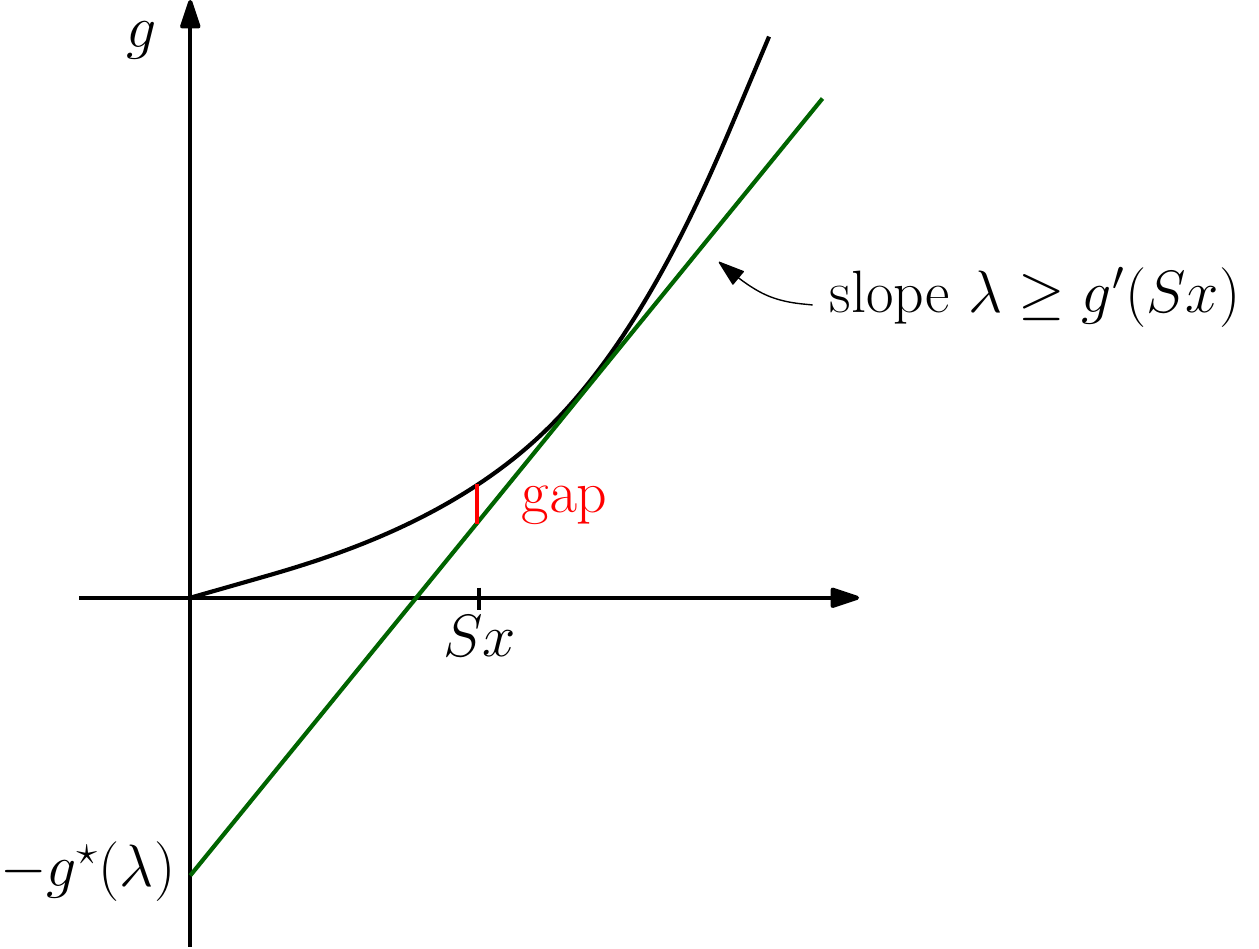}
  \caption{The overshoot gap is upper bounded by $\gs(\lambda)$.}
  \label{fig:offshoot}
\end{figure}

\section{Facts about Convex Functions and Useful Inequalities}
\label{sec:app-facts}

\section{Missing Proofs} \label{app:proofs}

\subsection{Proofs for Section \ref{sec:superm-funct}} \label{app:supermod}

\begin{proof}[Proof of Lemma \ref{lem:superadd2}]
	By integrating gradients we have  
  \[ f(x'+y) - f(x') = \int_{t = 0}^1 \ip{ \nabla f(x' + t\, y), y } \, dt \leq \int_{t
      = 0}^1 \ip{ \nabla f(x + t\, y), y } \, dt = f(x+y) - f(x), \] where the
  inequality uses Lemma \ref{lemma:supermod}(1), the monotone gradients property. 
\end{proof}

\begin{proof}[Proof of Fact \ref{fct:supermod}]
  Given $x' \leq x$ in the Boolean cube and $i$ not in $x$,
  \[ h(x' + \be_i) - h(x') = g(S(x' + \be_i)) - g(Sx') \leq g(S(x +
  \be_i)) - g(Sx) = h(x + \be_i) - h(x), \] where the inequality uses
  Lemma~\ref{lem:superadd2}, and $S\ge 0$.
\end{proof}


\subsection{Proof of Lemma \ref{lemma:piplus}}

	Since $g$ is separable, it has the form $g(z) = \sum_i g_i(z_i)$. Thus, it will suffice to perform the gradient truncation on each of the 1-dimensional functions $g_i$, which is accomplished by the following lemma.



\begin{claim}[Gradient truncation] 
  \label{clm:grad-trunc}
  Given a $1$-dimensional convex function $f: \R \to \R$ and a real
  value $\gamma \in \R_+$, there is another function
  $f^+ : \R \rightarrow \R$ satisfying the following:
  \begin{OneLiners}
  \item[i.] $f^+$ is convex,
  \item[ii.] For all $z$, all subgradients $u \in \partial f^+(z)$, and
    satisfy $u \le \gamma$.
  \item[iii.] If $z$ is such that some subgradient $z \in \partial f(z)$
    satisfies $u \le \gamma$, then $f^+(z) = f(z)$.
  \end{OneLiners}
\end{claim}
	
\begin{subproof}
  Define $f^+$ as follows: 
  \begin{align*}
    f^+(z) = \sup_{u \le \gamma} \left[ u\cdot z - \fs(u)
    \right];
  \end{align*}
  notice the constraint $u \le \gamma$, and that the dual $\fs$ is the
  usual Fenchel dual for $f$.  Properties~(i) and~(ii) follow from the
  fact $f^+$ is the point-wise supremum of linear functions with slopes
  at most $\gamma$.

  
  For Property~(iii), by the double dual property in
  Claim~\ref{clm:doubleDual}, $f(z) = \sup_{u} [ u\cdot z - \fs(u)],$ and
  hence $f \ge f^+$. Moreover, take a point $z$ such that some
  subgradient $u \in \partial f(z)$ satisfies $u \le \gamma$. Then we have
  %
  \begin{gather*}
    f(z) \stackrel{\text{(Claim~\ref{clm:linear})}}{=} \overline{u} \cdot z - \fs(u) \le \sup_{u \le \lambda^*}
    \, [ u\cdot z  - \fs(u)] \stackrel{\text{(by defn.)}}{=} f^+(z) \leq f(z).
  \end{gather*}
  This concludes the proof. 
\end{subproof}

	Now define $g^+_i$ as the function obtained by applying the truncation lemma above with $f = g_i$ and $\gamma = \lambda^*_i$. Also define the truncated cost function $g^+$ as $g^+(z) = \sum_i g^+_i(z_i)$, and the truncated profit function $\profit^+$ as $\profit^+(x) = \ip{v,x} - g^+(Sx)$.

	We claim that $\profit^+$ satisfies all properties required by the lemma. First, it is discrete submodular: $g^+$ is convex and
supermodular, since it is a sum of $1$-dimensional convex functions, which are trivially supermodular; discrete submodularity of $\profit^+$ then follows from Fact~\ref{fct:supermod}.

	Moreover, we claim $\profit^+$ has the same value as $\profit$ over solutions in $\F|_{\Aopen_{\lambda^*}}$. This follows from the fact that every solution $x$ in this family has $\nabla g(Sx) \le \lambda^*$ (by Property~(P1')), and hence  Claim~\ref{clm:grad-trunc}(iii) guarantees that $g(Sx) = g^+(Sx)$.
	
	Finally, we claim that $\profit^+$ is non-negative monotone. Since $\profit^+(\emptyset) = \profit(\emptyset) = 0$, it again suffices to show monotonicity. For that,  Claim~\ref{clm:grad-trunc}(ii) guarantees that all solutions $x \subseteq \Aopen_{\lambda^*}$ now satisfy
  $\gr g^+(Sx) \leq \lambda^*$. The proof of Lemma~\ref{lemma:mono-feas} then extends to show that $\profit^+$ is monotone.
 This concludes the proof of Lemma \ref{lemma:piplus}.  


\subsection{Proof of Claim \ref{clm:lip}}

  Say $A \setminus B = \{i\}$. Then by supermodularity of the cost
  function $g$ (Lemma~\ref{lem:superadd2}) and the absence of
  exceptional items,
  \[
    \profit(y_A) - \profit(y_B) = v_iy_i - (g(Sy_A) - g(Sy_B)) \le
    v_iy_i - g(S^iy_i) \leq \profit(\be_i) \le M.
  \] 
  For the lower bound, we also use the optimality of $y$. 
  \[
    \profit(y_A) - \profit(y_B) = v_iy_i - (g(Sy_A) - g(Sy_B)) \ge v_iy_i -
    (g(Sy_U) - g(Sy_{U \setminus \{i\}})) = \profit(y) - \profit(y_{U
      \setminus \{i\}})\ge 0,
  \]
  where the middle inequality uses supermodularity of $g$, and the
  last one uses the optimality of $y$. This concludes the proof.

  
  \subsection{Proof of Lemma \ref{lem:goldilocks}}
  
  Part (b) was already proved in details, so we provide only a proof for part (a). 
    
  Let $\FOPT' = \FOPT(U_{\lambda^*})$. Using Theorem~\ref{thm:constr-main} and the assumption that no item has profit more than $\nicefrac{\OPT}{\eta d}$, we know $\FOPT' \ge \OPT(\Aopen_{\lambda^*}) \ge \frac{\FOPT}{2d+1} \geq \frac{\FOPT}{3d}$.
  Let $y$ be an optimal fractional solution for $U_{\lambda^*}$, so that
 $\profit(y) = \FOPT' \ge \frac{\FOPT}{3d}$.  By
  downward-closedness of $\F$, $y_\Isample$ is a feasible fractional
  solution only containing items in $\Isample_{\lambda^*}$. Therefore,
  $\E[\FOPT(\Isample_{\lambda^*})] \geq \E[\profit(y_\Isample)] =
  \FOPT'/2$. Finally, using the concentration inequality of Fact
  \ref{fact:submod-conc} with $t=\FOPT/12d$ (where the Lipschitz-type
  condition is satisfied due to Claim~\ref{clm:lip}), we get
  \[
    \Pr\left(\left|\profit(y_\Isample) - \frac{\FOPT'}{2}\right| \le
      \frac{\FOPT}{12d}\right) \ge 1- \frac{2 M\,
      (\FOPT'/2)}{(\FOPT/12d)^2} \ge \frac{19}{20}.
  \]
  The last inequality follows using $M = \frac{\OPT}{\eta d} \le \frac{\FOPT}{10^4 d}$.
  Hence, w.p.\ at least $\nicefrac{19}{20}$,
  $\FOPT(\Isample_{\lambda^*}) \ge 
  \pi(y_\Isample) \geq \frac{\FOPT'}{2} - \frac{\FOPT}{12d} \geq
  \frac{\FOPT}{12d}$. This is at least $\frac{1}{12d}\,\FOPT(\Isample)$
  since $L\subset U$.
  %
  %
  Hence $\lambda^*$ is a candidate in the definition of $\lhat$, thus by
  maximality $\lhat \ge \lambda^*$, proving the part (a) of the lemma.
  

\subsection{Proof of Theorem \ref{thm:online-constr}}

  Recall the definition of modified profit function $\profit^+$ from
  Lemma~\ref{lemma:piplus}. Again, let $\calE$ be the event that
  Lemmas~\ref{lemma:online-good-lambda} and~\ref{lem:constr-oos} hold;
  by a union bound $\Pr[\calE] \geq \nicefrac{17}{20}$. We claim that
  under this event, the functions $\profit^+$ and $\profit$ coincide
  over feasible sets of $\IoosF_{\lhat}$. Indeed,
  $\lhat \ge {\lambda^*}$, so the elements
  $\IoosF_{\lhat} \sse U_{{\lambda^*}}$, and then Lemma~\ref{lemma:piplus}(iii)
  tells us that $\profit^+$ and $\profit$ agree on all feasible subsets
  of $U_{\lambda^*}$.
		
  Using Assumption~\ref{ass:submod-algo}, the algorithm $\submod$
  applied over $\IoosF_{\lhat}$ cannot distinguish between $\profit$ and
  $\profit^+$. Thus, the solution $\mathbf{X} \in \calF$ returned by our
  algorithm using profit function $\profit$
  is 
  the same as one returned by running algorithm $\submod$ over instance
  $\IoosF_{\lhat}$ with the \emph{non-negative monotone submodular}
  function $\profit^+$. This algorithm is $\alpha$-competitive, and moreover
  conditioning on the sample still leaves the out-of-sample items in
  {random order}, so the guarantee of the algorithm still holds. Hence,
  \begin{align*}
    \E[\profit(\mathbf{X}) \mid \calE] ~~=~~
    \E[\profit^+(\mathbf{X}) \mid \calE] ~~
    \stackrel{\textrm{$\alpha$-approx}}{\ge} ~~
    \frac{1}{\alpha}\,\E[\OPT(\Ioos_{\lhat}) \mid \calE],
  \end{align*}
  where the first equality follows from $\profit^+$ and $\profit$
  agreeing on $\Ioos_{\lhat}$. Since we are assuming no item has large profit, by Theorem~\ref{thm:constr-main},
  the integer optimum is at least a $\nicefrac{1}{(2d+1)}$-fraction of
  the fractional optimum,
  \begin{align*}
    \frac{1}{\alpha}\,\E[\OPT(\Ioos_{\lhat}) \mid \calE]  ~~\ge~~
    \frac{1}{\alpha(2d+1)}\,\E[\FOPT(\Ioos_{\lhat}) \mid \calE]  ~~\ge~~
    \frac{\FOPT}{200 \alpha d(2d+1)},
  \end{align*}
  the last inequality using that event $\calE$ guarantees
  Lemma~\ref{lem:constr-oos}(c).  Since the algorithm does not include
  items with negative marginals, it always produces solutions with
  non-negative values. Therefore,
  \begin{align*}
    \E[\profit(\mathbf{X})] \ge \E[\profit(\mathbf{X})
    \mid \calE]\, \Pr(\calE) \ge \frac{\OPT}{O(\alpha d^2)}. 
  \end{align*}
  The last inequality follows since $\Pr(\calE)$ is a constant. This
  concludes the proof.



\subsection{Proof of Theorem \ref{thm:reduction}}

  Let $x^*$ be the optimal solution with value $\OPT := \profit(x^*)$
  for the problem of maximizing profit with the supermodular cost
  function $g$.  Since our proof deals with fractional allocations,
  define $\widehat{v}(\cdot)$ to be the convex extension (or Lov\'asz
  extension) of the submodular value function
  $v(\cdot)$~\cite{Schrijver-book}. Since $\widehat{v}$ is an extension
  of $v$, $v(x) = \widehat{v}(x)$ for all points $x$ in the domain of
  $v$, i.e., for $x \in \{0,1\}^n$. Define for all $x \in [0,1]^n$,
  \[ \overline{\profit}(x) := \widehat{v}(x) - \overline{g}(Sx). \] 
  Now,
  \begin{gather}
    \overline{\profit}(x^*/d) = \widehat{v}(x^*/d) - \overline{g}(Sx^*/d)
    \geq \nicefrac1d \, ( v(x^*) - g(Sx^*) ) =
    \profit(x^*)/d = \OPT/d. \label{eq:profit-good}
  \end{gather}
  The inequality uses that for the fractional point $x^*/d$, the
  Lov\'asz extension value is
  $\widehat{v}(x^*/d) = (1-\nicefrac1d)\cdot v(\emptyset) + \nicefrac1d
  \cdot v(x^*)$, and that by Lemma~\ref{lemma:dapprox}(1),
  $\overline{g}(Sx^*/d) \leq g(Sx^*)$.

  So the separable problem has a good \emph{fractional} solution
  $x^*/d$, and we want to ``round'' it to a near-integral
  solution. Indeed, take the matroid polytope $\calP$ corresponding to
  the matroid constraint $\F$, and intersect $\calP$ with the subspace
  $\{x \mid Sx = S(x^*/d)\}$. Clearly $x^*/d$ belongs to this
  intersection. Now consider maximizing the linear function $\ip{\gr
    \widehat{v}(x^*/d), x - x^*/d}$ over this polytope, and let
  $\tilde{x}$ be a basic feasible solution to this linear optimization
  problem. Since at most $d$ of the tight constraints come from the
  subspace restriction, the point $\tilde{x}$ lies on some face of the
  matroid polytope of dimension at most $d$.
  By~\cite[Theorem~4.3]{GRSZ14}, $\tilde{x}$ has at most $2d$ fractional
  coordinates. Moreover, since $\tilde{x}$ is the maximizer of the
  linear function and $x^*/d$ is a feasible point, the inner product
  $\ip{\gr \widehat{v}(x^*/d), \tilde{x} - x^*/d} \geq 0$. The convexity
  of the Lov\'asz extension now implies $\widehat{v}(\tilde{x}) \geq
  \widehat{v}(x^*/d)$. Because $S\tilde{x} = Sx^*/d$, the cost remains
  unchanged and we get
  \[ \overline{\profit}(\tilde{x}) \geq \overline{\profit}(x^*/d) \geq
    \Opt/d. \]

  Let $x^{\text{int}}$ be the $\tilde{x}$ restricted to the integral
  coordinates, and let $F \sse [n]$ be the set of fractional coordinates
  in $\tilde{x}$. Then by subadditivity of the $\overline{\profit}$
  function, we get the following, where $\chi_e$ is an indicator vector of element $e$. 
  \begin{align*}
    \OPT/d &\leq \overline{\profit}(\tilde{x}) \leq
    \overline{\profit}(x^{\text{int}}) + \sum_{e \in F}
    \overline{\profit}(\tilde{x}_e\, \chi_e) \tag{by subadditivity of
      $\overline{\profit}$} 
    \\
    &\leq \overline{\profit}(x^{\text{int}}) + \sum_{i \in F}
    \big( \widehat{v}(\tilde{x}_e \, \chi_e) - g(\tilde{x}_e \, S\chi_e)
      \big) . \tag{Definition of $\overline{\profit}$, and Lemma~\ref{lemma:dapprox}(2)}
  \end{align*}
  Moreover, for each individual item $e$,
  \[ \widehat{v}(\tilde{x}_e \, \chi_e) - g(\tilde{x}_e \, S\chi_e) =
    \tilde{x}_e \,\widehat{v}( \chi_e) - g(\tilde{x}_e \, S\chi_e) \leq
    \widehat{v}( \chi_e) - g(S\chi_e) = \profit(\chi_e). \] The first
  equality above uses that the Lov\'asz extension acts linearly on
  single items. The inequality follows since there are no exceptional
  items and $\tilde{x}_e \in (0,1)$.
  Hence, we get
  $\overline{\profit}(x^{\text{int}}) + \sum_{e \in F} \profit(\chi_e)
  \geq \Opt/d$.

  We can use the algorithm for the separable problem (which is part of
  the theorem assumption) to find $x^{\text{sep}}$ with value
  $\overline{\profit}(x^{\text{sep}}) \geq (\nicefrac1{\beta})
  \overline{\profit}(x^{\text{int}})$. Using
  Lemma~\ref{lemma:dapprox}(2) again,
  $\profit(x^{\text{sep}}) \geq
  \overline{\profit}(x^{\text{sep}})$. Also, using the well-known 
  $1/\mathrm{e}$-approximation for the most profitable item returns an
  item $e^*$ with profit
  $\profit(\chi_{e^*}) \geq \frac1{\mathrm{e} \cdot 2d} \sum_{e \in F}
  \profit(\chi_e) $. Returning $x^{\text{sep}}$ with probability
  $\frac{\beta}{\beta + 2\mathrm{e} d}$ and the single element
  $\widehat{e}$ otherwise gives expected value at least
  \[ \frac{1}{\beta + 2\mathrm{e} d} \bigg( \beta \,\profit(x^{\text{sep}}) +
    2\mathrm{e} d\, \profit(\chi_{e^*}) \bigg) \geq \frac{1}{\beta +
    2\mathrm{e} d} \bigg( \profit(x^{\text{int}}) +  \sum_{e \in F}
    \profit(\chi_e) \bigg) \geq \frac{\OPT}{d(\beta + 2
    \mathrm{e} d)}. \qedhere \] 



\section{Other Loose Ends}
\label{sec:other-loose-ends}
\subsection{Conjugates over the Positive Orthant}
\label{sec:conj-positive}

The following lemma justifies why it is enough to consider only
non-negative $\lambda$s for our setting.

\begin{lemma}\label{lem:pos-lambda}
  Given a convex, non-decreasing, non-negative function
  $g: \R^d_+ \to \R_+$, $\forall z \in \R^d_+$, we have
  $g(z) = \max_{\lambda \in \R^d} (\ip{\lambda,z} - \gs(\lambda)) =
  \max_{\lambda \in \R^d_+} (\ip{\lambda,z} - \gs(\lambda))$.
\end{lemma}

\begin{proof}
  Let $\ghat$ be a function that is same as $g$ on positive orthant and
  is $\infty$ everywhere else. Then for $\lambda\in \R^d$,
  \[ \ghats (\lambda) = \sup_{z \in \R^d} \{\ip{\lambda,z} - \ghat(z)\}
    = \sup_{z\geq0} \{\ip{\lambda,z} - \ghat(z)\} =
    \gs(\max(\lambda,0)). \] Here, the first equality is by
  definition. The second is because $\ghat(z) = \infty$ if $z$ is not
  non-negative. The third is because if some coordinate of $\lambda$ is
  negative, zeroing out the corresponding coordinate of $z$ increases
  $\ip{\lambda, z}$ and decreases $\ghat(z)$, because $\ghat$ is
  non-decreasing in positive orthant. Here vector $\max(\lambda,0)$ is
  the coordinate-wise maximum.

  Now for any $z\in \R^d_+$,
  $g(z) = \ghat(z) = \max_{\lambda \in \R^d} (\ip{\lambda,z} -
  \ghats(\lambda)) = \max_{\lambda \in \R^d} (\ip{\lambda,z} -
  \gs(\max(\lambda,0))) = \max_{\lambda \in \R_+^d} (\ip{\lambda,z} -
  \gs(\lambda))$.
\end{proof}

\subsection{Removing Assumptions on the Elements}
\label{sec:except}

Let $\OPT$ denote the profit of the optimal integer solution to the
problem~(\ref{eq:2}).  To discharge the conditions on elements in
Assumption~\ref{asm:nice} we show the following reduction.

\begin{lemma}
  \label{lem:reduce-nice}
  Suppose $\calA$ is algorithm that works for instances that have no
  exceptional items, where each item has profit
  $\profit(e) \leq \OPT/\eta d$, and where items are in general
  position, that guarantees a profit of $\OPT/\beta$. Then we can get
  another algorithm that requires none of these assumptions, and
  guarantees a profit of $\frac{\OPT}{O(\beta + \eta d)}$.
\end{lemma}

\begin{proof}
  The general position argument is simplest: we essentially need that
  for some fixed $\lambda \in \CC$, there is at most one element such
  that $v(e) = \ip{\lambda, s(e)}$. This can be achieved by subtracting
  from each $v(e)$ some random noise picked uniformly from the interval
  $[0,(\delta/n) \profit(e)]$ for some tiny $\delta$; this can change the
  optimal profit most by a $(1-\delta)$-factor.
  
  Recall that item $e$ is called exceptional if
  $\arg\max_{\theta \in [0,1]}\big\{ \theta\, v(e) - g(\theta\, s(e))
  \}$ is achieved at $\theta \in (0,1)$: i.e., it is optimal to take a
  fraction of the item. E.g., in the $1$-dimensional case,
  $v(e) = s(e) = 1$, and $g(s) = 0.99s^2$. The following claim is a
  minor variation of~\cite[Lemma~5.1]{BUCM}:

\begin{claim}[Few Exceptional Items]
  \label{clm:few-except}
  If $g$ is supermodular, then any optimal solution contains at most $d$ exceptional items.
\end{claim}

\begin{proof}
  Fix an optimal solution $O^*$, and for each $i = 1,\ldots, d$, let
  $o_i := \arg\max_{o \in O^*} s(o)_i$ be an item for which the
  $i^{th}$-coordinate of the size vector is maximized. Let $L$ denote
  the set of these ``large'' items. If $O^*$ contains strictly more than
  $d$ exceptional items, let $o^* \in O^*$ be any exceptional item not
  in $L$, and let $x' := \chi_{O^* \setminus \{o^*\}}$ be the
  characteristic vector for the elements in the optimal set without
  $o^*$. By construction, $s(o^*) \leq Sx'$ component-wise. Since $o^*$
  is exceptional, $v(o^*) < \ip{ \gr g(s(o^*)), s(o^*) }$; moreover, the
  latter is at most $\ip{ \gr g(Sx'), s(o^*) }$ due to $g$ having
  monotone gradients. But this implies that dropping $o^*$ would
  increase the profit, which contradicts our choice of $O^*$.
\end{proof}

Moreover, since the profit function is subadditive, there can be at most
$\eta d$ high-valued items. Now the reduction procedure: with
probability $\frac12$ run the single-item secretary problem (with
competitive ratio $1/\mathrm{e}$), and with the remaining probability
run algorithm $\calA$. If the instance has a high-valued item then we
get expected value at least $(1/\mathrm{e}) \cdot \Opt/(\eta d)$. If
not, divide the optimal solution $x^*$ into the solution restricted to
the non-exceptional items $x^1$, and to the (at most $d$) exceptional
items $x^2$. By the subadditivity of the profit,
$\profit(x^*) \leq \profit(x^1) + \profit(x^2)$. Again the secretary
algorithm gives a $1/(d\mathrm{e})$-approximation for the profit
$\profit(x^2)$, so it suffices to get a good approximation for the
non-exceptional items.
\end{proof}

\end{document}